\documentclass[reqno]{amsart}
\usepackage{amsmath,amsthm,amsfonts,amssymb,srcltx}
\usepackage{graphicx} 
\usepackage{epstopdf}
\usepackage{enumerate}
\usepackage{color}

\newcommand{\bea}{\begin{eqnarray}}
\newcommand{\eea}{\end{eqnarray}}
\newcommand{\be}{\begin{equation}}
\newcommand{\ee}{\end{equation}}

\newtheorem{theorem}{Theorem}[section]
\newtheorem{proposition}[theorem]{Proposition}

\theoremstyle{definition}

\renewenvironment{proof}{{\noindent\bf Proof.}}{\hfill $\Box$\par\vskip3mm}

\newcommand{\ttt}{\pmb{t}}
\newcommand{\uuu}{\pmb{u}}

\newcommand{\Nd}{\mathcal{N}}
\newcommand{\Mcal}{\mathcal{M}}

\newcommand{\bb}{\pmb{b}}
\newcommand{\db}{\pmb{d}}
\newcommand{\Zb}{\mathbb Z}
\newcommand{\fb}{\pmb{f}}
\newcommand{\pb}{\pmb{p}}
\newcommand{\nb}{\pmb{n}}
\newcommand{\mb}{\pmb{m}}

\begin{document}

\date{}

\author[J. E. Andersen]{J{\o}rgen Ellegaard Andersen}
\address{QGM\\
Department of Mathematics\\
Aarhus University\\
DK-8000 Aarhus C\\
Denmark}
\email{jea.qgm{\char'100}gmail.com}

\author[H. Fuji]{Hiroyuki Fuji}
\address{Faculty of Education\\ 
Kagawa University\\
Takamatsu 760-8522\\
Japan
QGM\\
Aarhus University\\
DK-8000 Aarhus C\\
Denmark
}
\email{fuji{\char'100}ed.kagawa-u.ac.jp}

\author[R. C. Penner]{Robert C. Penner}
\address{Institut des Hautes {\'E}tudes Scientifiques, 35 route de Chartres, 91440 Burs-sur-Yvette, France;
Division of Physics, Mathematics and Astronomy, California Institute of Technology, Pasadena, CA 91125, USA}
\email{rpenner{\char'100}caltech.edu,\hspace{0.3cm}rpenner@ihes.fr}

\author[C. M. Reidys]{Christian M. Reidys}
\address{Biocomplexity Institute of Virginia Tech Blacksburg, VA 24061, USA}
\email{duckcr{\char'100}vbi.vt.edu}

\title [Partial chord diagrams and boundary length and point spectrum]{The boundary length and point spectrum enumeration of partial chord diagrams\\ using cut and join recursion}

\thanks{Acknowledgments: 
The authors thank Masahide Manabe and Piotr Su{\l}kowski  for useful comments.
JEA and RCP are supported by the Centre 
for Quantum Geometry of Moduli Spaces which is funded by the Danish National Research Foundation.
The research of HF is supported by the
Grant-in-Aid for Research Activity Start-up [\# 15H06453], Grant-in-Aid
for Scientific Research(C)  [\# 26400079], and Grant-in-Aid for Scientific
Research(B)  [\# 16H03927]  from the Japan Ministry of Education, Culture,
Sports, Science and Technology.
}
\begin{abstract}
We introduce the boundary length and point spectrum, as a joint generalization of the boundary length spectrum and boundary point spectrum in \cite{AAPZ}. We establish by cut-and-join methods that the number of partial chord diagrams filtered by the boundary length and point spectrum satisfies a recursion relation, which combined with an initial condition determines these numbers uniquely.
This recursion relation is equivalent to a 
second order, non-linear, algebraic partial differential
equation for  the generating function of the numbers of partial chord diagrams filtered by the boundary length and point spectrum.
\end{abstract}

\maketitle
\tableofcontents

\section{Introduction}\label{sec1}

A {\em partial chord diagram}, is a special kind of graph, which can be specified as follows. The graph consists of a number of line segments (which we will also call backbones) arranged along the real line (hence they come with an ordering) with a number of vertices on each. A number of semi-circles (called chords) arranged in the upper half plan are attached at a subset of the vertices of the line segments, in such a way that no two chords have endpoints on the line segments in common. The vertices which are not attached to chord ends are called the marked points. 
A {\em chord diagram} is by definition a partial chord diagram with no marked points. 
Partial chord diagrams occur in many branches of mathematics,
including topology \cite{Ba,Kontsevich2}, geometry \cite{AMR1,AMR2,ABMP}
and representation theory \cite{C-SM}. Furthermore, they play a very prominent role in macro molecular biology. 
Please see the introduction of \cite{AFMPS2} for a short review of these applications.

As documented in \cite{OZ,VOZ,VO, APRWat,APRWan,AHPR,ACPRS,ACPRS2,AAPZ,Reidys,P5,AFMPS1,AFMPS2}, 
the notion of a {\it fatgraph} \cite{P1,P2,P3,P4} is a useful concept when studying partial chord diagrams\footnote{
In \cite{MY,CY}, the Schwinger-Dyson approach to the enumeration of chord diagrams is also discussed.
}.
A fatgraph is a graph together with a cyclic ordering on each collection of half-edges 
incident on a common vertex.
A partial linear chord diagram $c$ has a natural fatgraph structure induced from its presentation in the plane.
The fatgraph $c$ has canonically a two dimensional surface with boundary $\Sigma_c$ associated to it
(e.g. see Figure \ref{partial_chord}).

\begin{figure}[h]
\begin{center}
   \includegraphics[width=120mm,clip]{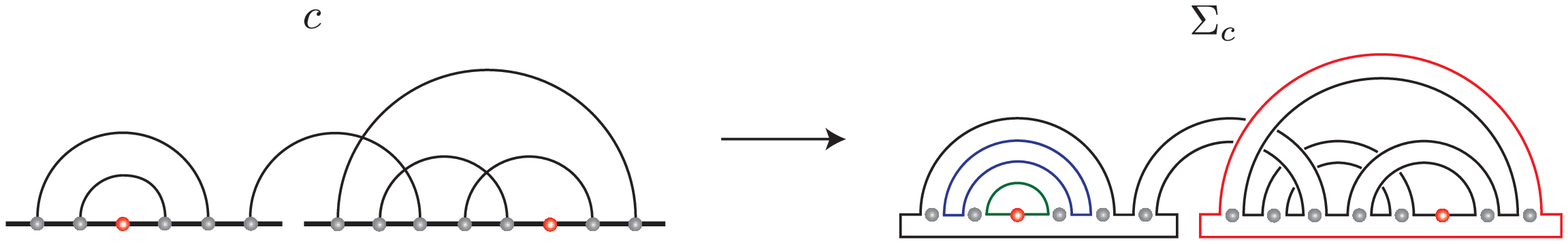}
\end{center}
\caption{\label{partial_chord} The partial chord diagram $c$ and the surface $\Sigma_c$ associated to the fatgraph
 with marked points. 
 This partial chord diagram has the type
 $\{g,k,l;\{b_i\};\{n_i\};\{\ell_i\}\}=\{1,6,2;\{b_6=1,b_8=1\};\{n_0=2,n_1=2\};\{\ell_1=1,\ell_2=2,\ell_9=1\}\}$. The boundary length-point spectra are $\{m_{(1)}=1,m_{(0,0)}=2,m_{(0,0,0,0,0,1,0,0,0)}=1\}$.
 }
\end{figure}

We now recall the basic definitions from \cite{AAPZ} for a partial chord diagram $c$. 
\begin{itemize}
\item The number  of chords, the number of marked points, and the number of backbones of $c$ are denoted $k$,  $l$, and $b$ respectively.

\item 
The Euler characteristic and the genus of $\Sigma_c$, are denoted $\chi$ and $g$ respectively. If $n$ is the number of boundary components of $\Sigma_c$, we have that
\begin{align}
\chi=2-2g,
\end{align}
and $g$ obeys Euler's relation
\begin{align} 
2-2g=b-k+n.
\end{align}

\item The  \textit{backbone  spectrum} $\bb=(b_0,b_1,b_2,\ldots)$ are assigned to $c$, if it has $b_i$ backbones with precisely $i\geq 0$ vertices (of  degree either two or three);\\

\item The \textit{boundary point spectrum} $\nb=(n_0,n_1,\ldots)$ is assigned to $c$, if its boundary contains $n_i$ connected components with $i$ marked points;\\

\item The \textit{boundary length spectrum} ${\boldsymbol\ell}=(\ell_1,\ell_2,\ldots )$ is assigned to $c$,  if the boundary cycles of the diagram
consist of $\ell_K$ edge-paths of length $K\geq 1$, where the {\em length} of a boundary cycle is the number of chords
it traverses counted with multiplicity  (as usual on the graph obtained from the diagram by collapsing each backbone to a distinct point) {\sl plus} the number of backbone undersides it traverses (or in other words, the number of traversed connected components obtained by removing all the chord endpoints from all the backbones).\\

\end{itemize}

We now introduce the combination of the boundary length spectrum and the boundary point spectrum, namely our new boundary length and point spectrum.

\begin{itemize}

\item The \textit{boundary length and point spectrum} $\pmb{m}=(m_{(d_1,\ldots,d_K)})$  is assigned to $c$,  if its boundary contains $m_{(d_1,\ldots,d_K)}$ connected components of length $K$ with marked point spectrum $(d_1,\ldots, d_K)$, meaning that there cyclically around the boundary components are $d_1$ marked points, then a chord or a backbone underside, then $d_2$ marked points, then a chord or a backbone underside, and so on all the way around the boundary component. In fact we will not need to distinguish which way around the boundary we go. Hence it is only the cyclic ordered tuple of the numbers $d_1, \ldots, d_K$, which we need and which we denote as $\db_K = (d_1,\ldots,d_K)$. We remark that some of the $d_I$ ($1\le I\le K$) might be zero.\\

\end{itemize}

We have the following relations
\begin{align}
& b=\sum_{i\geq 0} b_i, \; n=\sum_{i\geq 0} n_i=\sum_{K\geq 1} \ell_K =  \sum_{K\geq 1} \sum_{\db_K}  m_{\db_K},  
\nonumber \\
& 2k+l=\sum_{i>0} ib_i,\; l=\sum_{i>0} in_i =  \sum_{K\geq 1} \sum_{\db_K}  |\db_K|m_{\db_K}
\nonumber \\
& 2k+b=\sum_{K\geq 1}K\ell_K = \sum_{K\geq 1} \sum_{\db_K}  Km_{\db_K},
\nonumber 
\end{align}
 where $|\db_K| = \sum_{I=1}^K d_I$. For all $K$ and $i$, we also have that
 \begin{align}
\ell_K =\sum_{\db_K}  m_{\db_K},\quad n_i = \sum_{K\ge 1} \sum_{i= |\db_K|} m_{\db_K}. 
\nonumber
\end{align}
We define $\Mcal_{g,k,l}({\bb},{\mb})$  to be the number of connected partial chord diagrams of type $\{g,k,l;\bb;\mb\}$ taken to be zero if there is none of the specified type. 
In \cite{AAPZ}, 
$\mathcal{N}_{g,k,l}(\bb,\nb,\pb)$ is defined as the number of distinct connected partial chord diagrams
of type $\{g, k, l; \bb;\nb; \pb\}$.
We find the relation between these numbers by the following formula
\begin{align}
\Nd_{g,k,l}(\bb,\boldsymbol\ell,\nb) = \sum_{\mb \in M(\boldsymbol\ell,\nb)}  \Mcal_{g,k,l}(\bb,\mb),
\nonumber
\end{align}
where
\begin{align}
M(\boldsymbol\ell,\nb) = \bigl\{ \mb \,\big|\, \ell_K =\sum_{\db_K}  m_{\db_K}\mbox{, } n_i = \sum_{K\ge 1} \sum_{i = |\db_K|} m_{\db_K} \bigr\}.
\nonumber
\end{align}
In particular, the numbers $\Nd_{g,k,l}(\bb,\nb)$ and $\Nd_{g,k,b}(\boldsymbol\ell)$ are given by
\begin{align}
\Nd_{g,k,l}(\bb,\nb)=\sum_{\boldsymbol\ell} ~\Nd_{g,k,l}(\bb,\boldsymbol\ell,\nb),
\quad
\Nd_{g,k,b}(\boldsymbol\ell)=\sum_{\nb} \sum_{\sum b_i=b} ~\Nd_{g,k,l=0}(\bb,\boldsymbol\ell,\nb),
\nonumber
\end{align}

For the index $\bb=(b_{i})$, we consider the variable $\pmb{t}=(t_i)$ and denote
\begin{align}
\ttt^{\pmb{b}}=\prod_{i\ge 0}t_i^{b_i}.
\nonumber
\end{align}
And for the index $\db= (\db_{K})$,
we consider the variable $\uuu=(u_{\db_K})$
and denote
\begin{align}
\uuu^{\pmb{m}} = \prod_{K\ge 1} \prod_{\db_K} u_{\db_K}^{m_{\db_K}}
\nonumber 
\end{align}
for any $\mb = (\mb_{\db_K})$.
We define the orientable, multi-backbone, boundary length and point spectrum generating function $H(x,y;\pmb{t};\pmb{u})=\sum_{b\geq 0} \mathcal{F}_b(x,y;\pmb{t};\pmb{u})$, where
\begin{align}\label{gf}
H_b(x,y;\pmb{t};\pmb{u})=
{1\over{b!}} \sum_{g=0}^\infty \; \sum_{k=2g+b-1}^\infty \: \sum_{\substack{\sum_K\sum_{\db_K} m_{\db_K}\\=k-2g-b+2}}\;\sum_{\sum b_i=b}\Mcal_{g,k,l}(\bb,\mb) x^{2g} y^k{\pmb t}^{\bb} {\pmb{ u}}^{\mb} ,
\end{align}

For an element $\pb = (p_{({d_1},\ldots {d_K})})$, where each $p_{({d_1},\ldots {d_K})}\in \Zb$, we write
\begin{align}
\pb = \pb^+ -\pb^-,
\nonumber
\end{align}
where $\pb^+$ contains all the positive entries and $\pb^-$ the absolute value of all the negative ones, which we assume to both be finite. We define the differential operator 
\begin{align}
D_{\pb} = \prod_{\db} u_{\db}^{\pb^-_{\db}} \prod_{\db} \left( \frac{\partial}{\partial u_{\db}}\right)^{\pb^+_{\db}}.
\nonumber
\end{align}

We now define $s_{I,J,\ell,m}(\db_K)$, $s_{I,\ell,m}(\db_K)$ and $q_{I,J,\ell,m}(\db_K,\fb_L)$ to be strings like $\pb$ given by the following formulae
\begin{align}
%&s_{I,J,\ell,m}(\db_K) = \pmb{e}_{\db_K} - \pmb{e}_{(d_I -\ell-1, d_{I+1}, \ldots, d_{J-1}, m)} - \pmb{e}_{(d_J-m-1, d_{J+1}, \ldots, d_J,d_1,\ldots, d_{J-1}, \ell)},
&s_{I,J,\ell,m}(\db_K) = \pmb{e}_{\db_K}- \pmb{e}_{(d_1,\ldots,d_{I-1},d_{I}-\ell-1,m,d_{J+1},\ldots,d_{K})} - \pmb{e}_{(\ell,d_{I+1},\ldots,d_{J-1},d_{J}-m-1)},\nonumber\\
&s_{I,\ell,m}(\db_K) = \pmb{e}_{\db_K} - \pmb{e}_{(d_1,\ldots, d_{I-1},\ell,m,d_{I+1}, \ldots, d_K)} - \pmb{e}_{(d_I -\ell-m-2)} ,
\nonumber
\\
&q_{I,J,\ell,m}(\db_K,\fb_L) 
\nonumber \\
&= \pmb{e}_{\db_K} +\pmb{e}_{\fb_L}- \pmb{e}_{(d_1,\ldots, d_{I-1},d_{I}-\ell-1,m,f_{J+1},\ldots, f_L,f_1,\ldots,f_{J-1},f_{J}-m-1, \ell, d_{I+1},\ldots, d_K)}.
\nonumber
\end{align}
where $\pmb{e}_{\db_K}$ denotes the sequence $(0,\ldots,0,1,0,\ldots)$ where the component $1$ appears only at the entry indexed by $\db_K$.
We further define the index $c_{I,J,\ell,h}(\db_K,\fb_M)$ by the formula
\begin{align}
&c_{I,J,\ell,m}(\db_K,\fb_L) \nonumber\\
&= (d_1,\ldots, d_{I-1},d_{I}-\ell-1,m,f_{J+1},\ldots, f_L,f_1,\ldots,f_{J-1},f_{J}-m-1, \ell, d_{I+1},\ldots, d_K),
\nonumber
\end{align}
which is identical to the index on the last term of the above assignments.

\begin{theorem}[Enumeration of partial chord diagrams filtered by their boundary length and point spectrum] \label{thm1} ~\vskip .1in
 \noindent Define the first and second order linear differential operators
\begin{align}
M_{0}&=  \sum_{K\ge 1} \sum_{\db_K} \left( \sum_{1\leq J<I\leq K} \sum_{\ell=0}^{d_I-1}  \sum_{m=0}^{d_J-1} D_{s_{I,J,\ell,m}(\db_K)} + \sum_{ I=1}^K \sum_{\ell,m=0}^{d_I-1}  D_{s_{I,\ell,m}(\db_K)}\right),
\label{M0}\\
M_{2}&=\frac12 \sum_{K,L\ge 1} \sum_{\db_K, \fb_L}  \sum_{I=1}^K \sum_{J=1}^L  \sum_{\ell=0}^{d_I-1}  \sum_{m=0}^{d_J-1}   D_{q_{I,J,\ell,h}(\db_K)},  \label{M2}
\end{align}
and the quadratic differential operator
\begin{align}
S(H)=\frac12 \sum_{K,L\ge 1} \sum_{\db_K, \fb_L}  \sum_{I=1}^K \sum_{J=1}^L  \sum_{\ell=0}^{d_I-1}  \sum_{m=0}^{f_L-1}   u_{c_{I,J,\ell,m}(\db_K, \fb_L)} D_{\db_K}(H)  D_{\fb_L}(H) \;.
\label{S}
\end{align}
Then the following partial differential equations hold
\begin{align}
{{\partial H_1}\over{\partial y}} = (M_0+x^2M_2)H_1,\quad
{{\partial H}\over{\partial y}} = (M_0+x^2M_2+S)H.
\label{PDE_ori}
\end{align}
Together with the initial conditions  
\begin{align}
H_1(x,y=0;\ttt=(t_1);\uuu)=u_{(0)}t_{1},\quad 
H(x,y=0;\ttt;\uuu)=\sum_{i\geq 1} u_{(i)}t_{i},
\label{initial}
\end{align}
they determine the functions $H_1$ and $H$ uniquely.
\end{theorem}

In this article, we also consider the non-oriented analogue of partial chord diagrams.
The generalization of the above analysis is straightforward, as we will now explain.
A {\em non-oriented} partial chord diagrams, is a partial chord diagram together with a decoration of a binary variable at each chord, which indicates if the chord is {\em twisted} or not. When associating the surface ${\Sigma}_c$, to a non-oriented partial chord diagram, 
a twisted band is associated along twisted chords as indicated in 
Figure \ref{partial_chord_non}.
By this construction, $2^k$ orientable and non-orientable surfaces are obtained from one partial chord diagram with $k$ chords, when we vary over all assignments of twisting or not to the $k$ chords. 
In the non-oriented case, we have the following definition of the Euler characteristic.

\begin{itemize}
\item Euler characteristic $\chi$.\\
The Euler characteristic of the two dimensional surface $\Sigma_c$
is defined by the formula
\begin{align}
\chi=2-h,
\end{align}
where  $h$ is %half
 the number of cross-caps and we have Euler's relation
\begin{align}
 2-h=b-k+n.
\end{align}
\end{itemize}
With this set-up,
the enumeration of the non-oriented partial chord diagrams is considered in parallel to the oriented case discussed above with a small change for the boundary length and point spectrum $\mb$. In this non-oriented case, there are now induced orientation on the boundaries of $\Sigma_c$ and hence for an index $\db_K = (d_1,\ldots,d_K)$ corresponding some boundary component of $\Sigma_c$, we not only need to consider this tuple up to cyclic permutation of the tuple, but also reversal of the order
$$ \db_K = (d_1,d_2,\ldots,d_K) =  (d_K,\ldots d_2,d_1).$$

\begin{figure}[h]
\begin{center}
   \includegraphics[width=120mm,clip]{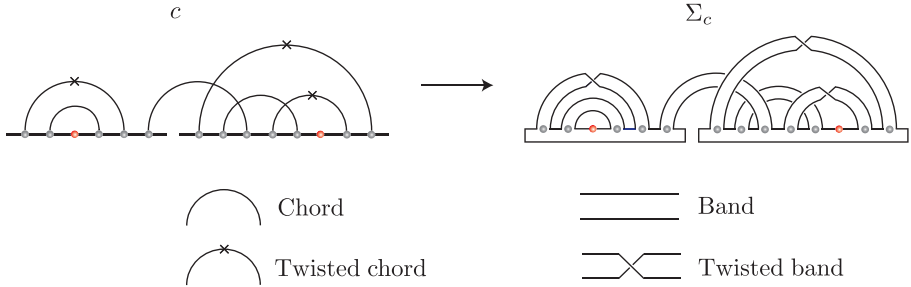}
\end{center}
\caption{\label{partial_chord_non} The non-oriented surface constructed out of untwisted and twisted chords.}
\end{figure}

Let $\widetilde\Mcal_{h,k,l}(\bb,\mb)$ be the number of non-oriented partial chord diagrams of type $\{h,k,l;\bb;\mb\}$.
In \cite{AAPZ},  $\widetilde\Nd_{h,k,l}(\bb,\boldsymbol\ell,\nb)$
is defined as 
the number of non-oriented connected partial chord diagrams of type $\{h,k,l;\bb;\boldsymbol\ell;\nb\}$.
These numbers are related by the following formula 
\begin{align}
\widetilde\Nd_{h,k,l}(\bb,\boldsymbol\ell,\nb) = \sum_{\mb \in M(\boldsymbol\ell,\nb)} \widetilde\Mcal_{h,k,l}(\bb,\mb),
\nonumber
\end{align}
and the numbers $\widetilde\Nd_{h,k,l}(\bb,\nb)$ and $\widetilde{\Nd}_{h,k,b}(\boldsymbol\ell)$ are given by
\begin{align}
&\widetilde\Nd_{h,k,l}(\bb,\nb)=\sum_{\boldsymbol\ell}~ \widetilde\Nd_{h,k,l}(\bb,\boldsymbol\ell,\nb),
\quad
&\widetilde{\Nd}_{h,k,b}(\boldsymbol\ell)=\sum_{\nb} \sum _{\sum b_i=b} ~\widetilde{\Nd}_{h,k,l=0}(\bb,\boldsymbol\ell,\nb).
\nonumber
\end{align}

We define the non-oriented generating function
$\widetilde{H}(x,y;\pmb{t};\pmb{u})=\sum_{b\geq 1} \widetilde{H}_b(x,y;\pmb{t};\pmb{u})$
to be given by
\begin{align}\label{gf_non}
\widetilde{H}_b(x,y;\pmb{t};\pmb{u})=
{1\over b!} \sum_{h=0}^\infty \; \sum_{k=h+b-1}^\infty \: \sum_{\substack{\sum_K\sum_{\db_K} m_{\db_K}\\=k-h-b+2}}\;\sum_{\sum b_i=b}\widetilde{\Mcal}_{h,k,l}(\bb,\mb) x^{h} y^k{\pmb t}^{\bb} {\pmb{ u}}^{\mb}.
\end{align}

We define $s^{\times}_{I,J,\ell,h}(\db_K)$, $s^{\times}_{I,\ell,h}(\db_K)$ and $q^{\times}_{I,J,\ell,h}(\db_K,\fb_L)$ to be by
\begin{align}
&s^{\times}_{I,J,\ell,m}(\db_K)  = \pmb{e}_{\db_K} - \pmb{e}_{(d_1, \ldots, d_{I-1}, \ell,m, d_{J-1}, \ldots, d_{I+1}, d_J -\ell-1, d_J -m-1, d_{J+1}, \ldots, d_K)}, \nonumber\\
&s^{\times}_{I,\ell,m}(\db_K)  = \pmb{e}_{\db_K} - \pmb{e}_{(d_1,\ldots, d_{I-1}, \ell, d_I-\ell-m-2, m, d_{I+1},\ldots, d_K)},
\nonumber\\
&q^{\times}_{I,J,\ell,m}(\db_K,\fb_L) 
\nonumber \\
&= \pmb{e}_{\db_K}  +\pmb{e}_{\fb_M}- \pmb{e}_{(f_{1},\ldots, f_{J-1},f_{J}-m-1,\ell, d_{I-1},\ldots, d_1,d_K,\ldots, d_{I+1}, d_I -\ell-1,m, f_{J+1},\ldots, f_L)}. \nonumber
\end{align}
And we also define indices $c^{\times}_{I,J,\ell,h}(\db_K,\fb_M)$ by the formula
\begin{align}
&c^{\times}_{I,J,\ell,h}(\db_K,\fb_L) \nonumber \\
&= (f_{1},\ldots, f_{J-1},f_{J}-m-1,\ell, d_{I-1},\ldots, d_1,d_K,\ldots, d_{I+1}, d_I -\ell-1,m, f_{J+1},\ldots, f_L),
\nonumber
\end{align}
which again, we note is identical to the index on the last term of the above assignments.

\begin{theorem} [Enumeration of non-oriented partial chord diagrams filtered by their boundary length and point spectrum] \label{thm2} ~\vskip .1in
 \noindent Define the first and second order linear differential operators
\begin{align}
M_{1}^{\times}&=  \sum_{K\ge 1} \sum_{\db_K} \left( \sum_{1\leq J<I\leq K} \sum_{\ell=0}^{d_I-1}  \sum_{m=0}^{d_J-1} D_{s^{\times}_{I,J,\ell,m}(\db_K)} + \sum_{ I=1}^K \sum_{\ell,m=0}^{d_I-1}  D_{s^{\times}_{I,\ell,m}(\db_K)}\right), \label{M1_non}
\\
M_{2}^{\times}&=\frac12 \sum_{K,L\ge 1} \sum_{\db_K, \fb_L}  \sum_{I=1}^K \sum_{J=1}^L  \sum_{\ell=0}^{d_I-1}  \sum_{m=0}^{d_J-1}   D_{q^{\times}_{I,J,\ell,m}(\db_K)},  
\label{M2_non}
\end{align}
and the quadratic differential operator
\begin{align}
S^{\times} (H)=\frac12 \sum_{K,L\ge 1} \sum_{\db_K, \fb_L}  \sum_{I=1}^K \sum_{J=1}^L  \sum_{\ell=0}^{d_I-1}  \sum_{m=0}^{f_L-1}   u_{c^{\times}_{I,J,\ell,m}(\db_K, \fb_L)} D_{\db_K}(H)  D_{\fb_L}(H) \;.
\label{S_non}
\end{align}
Then the following partial differential equations hold
\begin{align}
&{{\partial \widetilde H_1}\over{\partial y}} = (M_0+ xM^{\times}_1+x^2(M_2+M^\times_2))\widetilde H_1,
\nonumber \\
&{{\partial \widetilde H}\over{\partial y}} = (M_0+ xM^{\times}_1+x^2(M_2+M^\times_2) + S + S^\times)\widetilde H.
\label{pde_non}
\end{align}
Together with the following initial conditions 
\begin{align}
\widetilde{H}_1(x,y=0;\ttt=(t_1);\uuu)=u_{(0)}t_{1},\quad 
\widetilde{H}(x,y=0;\ttt;\uuu)=\sum_{i\geq 1} u_{(i)}t_{i},
\label{initial_non}
\end{align}
determines $\widetilde{H}_1$ and $\widetilde{H}$ uniquely. 
\end{theorem}

This paper is organized as follows. Section \ref{sec2} contains basic combinatorial results
on the boundary length and point spectra of partial chord diagrams and derives the recursion relation of the number of diagrams (Proposition \ref{prop1}), by the cut-and-join method. This cut-and-join equation is rewritten as a second order, non-linear, algebraic partial differential equation for  generating function of the number of partial chord diagrams filtered by the boundary length and point spectrum (Proposition \ref{prop2}).
Section \ref{sec3} extends these results to include the non-oriented analogues of the partial chord diagrams. The cut-and-join equation is extended to provide a recursion on the number of non-oriented partial chord diagrams (Proposition \ref{prop3}), and is also rewritten as  partial differential equation (Proposition \ref{prop4}).

%%%%%%%%%%%%%%%%%%%%%%%%%%%%%%%%%%%%%%%%%%%%%%%
%Section 2: Cut-and-join equation for the orientable diagrams
%%%%%%%%%%%%%%%%%%%%%%%%%%%%%%%%%%%%%%%%%%%%%%%
\section{Combinatorial proof of the cut-and-join equation}\label{sec2}
In this section, we devote to prove Theorem \ref{thm1}.
The partial differential equation (\ref{PDE_ori}) is equivalent to the following recursion relation for the numbers of connected partial chord diagrams. 
\begin{proposition}\label{prop1}
The numbers ${\mathcal M}_{g,k,l}(\bb,\mb)$ enumerating connected partial chord diagrams of type $\{g,k,l;\bb,\mb\}$
obey the following recursion relation
\begin{align}
&k\mathcal{M}_{g,k,l}(\bb,\mb)
\nonumber\\
=&\sum_{K\ge 1}\sum_{\db_K}
(m_{\db_K}+1)\biggl[
\sum_{1\le I<J\le K}\sum_{\ell=0}^{d_I-1}\sum_{m=0}^{d_J-1}
{\mathcal M}_{g,k-1,l+2}\left(\bb,\mb+s_{I,J,\ell,m}(\db_K)\right)
\nonumber \\
&\hspace*{3cm}
+\sum_{I=1}^K\sum_{\substack{
\ell,m=0\\
\ell+m\le d_I-2}}^{d_I-1}
{\mathcal M}_{g,k-1,l+2}\left(\bb,\mb+s_{I,\ell,m}(\db_K)\right)
\biggr]
\nonumber \\
&+\frac{1}{2}\sum_{K\ge 1}\sum_{L\ge 1}\sum_{\db_K}\sum_{\fb_L}
(m_{\db_K}+1)(m_{\fb_L}+1-\delta_{\db_K,\fb_L})
\nonumber \\
&\quad\times\sum_{I=1}^K\sum_{J=1}^L\sum_{\ell=0}^{d_I-1}\sum_{m=0}^{f_J-1}
{\mathcal M}_{g-1,k-1,l+2}\left(\bb,\mb+q_{I,J,\ell,m}(\db_K,\fb_L)\right)
\nonumber \\
&+\frac{1}{2}\sum_{K\ge 1}\sum_{L\ge 1}\sum_{\db_K}\sum_{\fb_L}
\sum_{g_1+g_2=g}\;\;\sum_{k_1+k_2=k-1}\;\;
\sum_{b^{(1)}+b^{(2)}=b}
\nonumber \\
&\quad\times
\sum_{I=1}^K\sum_{J=1}^L\sum_{\ell=0}^{d_I-1}\sum_{m=0}^{f_J-1}
\sum_{
\substack{
\mb^{(1)}+\mb^{(2)} \\
=\mb+q_{I,J,\ell,m}(\db_K,\fb_L)
}}
\nonumber \\
&\quad\quad\times
m_{\db_K}^{(1)}m_{\fb_L}^{(2)}
\frac{b!}{b^{(1)}!b^{(2)}!}
{\mathcal M}_{g_1,k_1,l_1}\bigl(\bb^{(1)},\mb^{(1)}\bigr)
{\mathcal M}_{g_2,k_2,l_2}\bigl(\bb^{(2)},\mb^{(2)}\bigr).
\label{cut_join}
\end{align}
\end{proposition}
This recursion relation is referred to as the \textit{cut-and-join equation}, since it follows from a cut-and-join argument, which we shall now provide.

\begin{proof}
When one removes one chord from a partial chord diagram, there are essentially three distinct possible outcomes. First of all the diagram can stay connected and then there are two cases to consider. In the first one, the chord that is removed is adjacent to two different boundary components and in the second one it is adjacent to just one. The third case is when the chord diagram becomes disconnected.

In the first case, 
the genus of the partial chord diagram is not changed, but two boundary components join into one component. 
On the other hand, in the second case, the genus decreases by one, and one boundary component splits into two components.

\begin{figure}[h]
\begin{center}
   \includegraphics[width=120mm,clip]{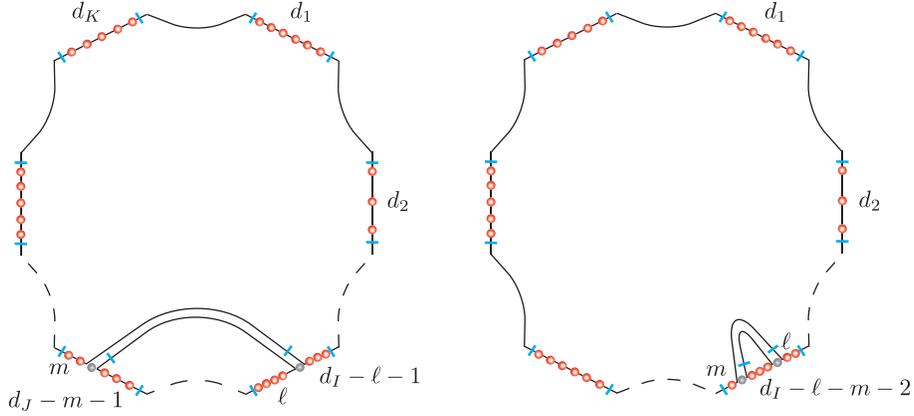}
\end{center}
\caption{\label{comb1} Removal of a chord in case one. The chord is depicted as a band. After the removal of this chord, two boundary components join into one component. 
 Left: The clusters of marked points $(d_I-\ell-1,m)$ and $(d_J-m-1,\ell)$  join into two clusters $d_I$ and $d_J$
  Right:  The clusters of marked points $(\ell,m)$ and $(d_J-m-1)$  join into one cluster $d_I$.}
\end{figure}
In the first case, and let us say that after removing this chord, the two adjacent boundary components join into one component with the marked point spectrum $\db_K=(d_1,\ldots,d_K)$.
 (See Figure \ref{comb1}.)
Under this elimination, the numbers $k$ and $n$ change to $k-1$ and $n-1$,  the genus $g$ is not changed (c.f. Euler's relation $2-2g=b-k+n$). The number of marked points $l$ changes to $l+2$, because the chord ends of the chord which is removed become new marked points.
There are two distinct possible sub cases, namely either the chord ends belong to two distinct clusters  of marked points $d_I$ and $d_J$ in the resulting chord diagram, or chord ends belong to the same cluster of  marked points $d_I$.

We will consider  the former kind of chord, and assume $I<J$ without loss of generality. 
Before we remove the chord, the two boundaries adjacent to the chord needs to have the following two marked point spectra
\begin{align}
&(d_1,\ldots,d_{I-1},d_{I}-\ell-1,m,d_{J+1},\ldots,d_{K}),
 \;\;\mathrm{and}\;\;
 (\ell,d_{I+1},\ldots,d_{J-1},d_{J}-m-1),
\nonumber \\
&0\le \ell\le d_I-1,\quad 0\le m\le d_J-1,
\nonumber
\end{align}
When removing the chord, we connect the clusters of marked points $(d_{I}-\ell-1,m)$ and $(d_{J}-m-1,\ell)$.
If the original partial chord diagram has the boundary length-point spectrum $\mb$,
the resulting diagram has 
\begin{align}
&\mb- \pmb{e}_{(d_1,\ldots,d_{I-1},d_{I}-\ell-1,m,d_{J+1},\ldots,d_{K})} - \pmb{e}_{(\ell,d_{I+1},\ldots,d_{J-1},d_{J}-m-1)} + \pmb{e}_{\db_K} 
\nonumber \\
&=\mb+s_{I,J,\ell,m}(\db_K).
\nonumber
\end{align}

For the latter kind, 
we must have two boundary components with the marked point spectra
\begin{align}
&(d_1,\ldots,d_{I-1},\ell,m,d_{I+1},\ldots,d_{K}),\;\;\mathrm{and}\;\; (d_{I}-\ell-m-2).
\nonumber \\
&0\le \ell,m\le d_I-1,\quad 0 \le \ell+m \le d_I-2, 
\nonumber
\end{align}
and removing the chord connects the clusters of marked points $(\ell,m)$ and $(d_{J}-m-1)$.
This manipulation changes the boundary length and point spectrum $\mb$ into
\begin{align}
\mb- \pmb{e}_{(d_1,\ldots, d_{I-1},\ell,m,d_{I+1}, \ldots, d_K)}  - \pmb{e}_{(d_I -\ell-m-2)} + \pmb{e}_{\db_K}
= \mb+s_{I,\ell,m}(\db_K).
\nonumber
\end{align}

For both of these two kinds of removal, there are $m_{\db_K}+1$ possibilities to choose the boundary components in the partial chord diagram.
Therefore, the number of possibilities for the first way of removal is
\begin{align}
&\sum_{K\ge 1}\sum_{\db_K}
(m_{\db_K}+1)\biggl[
\sum_{1\le I<J\le K}\sum_{\ell=0}^{d_I-1}\sum_{m=0}^{d_J-1}
{\mathcal M}_{g,k-1,l+2}\left(\bb,\mb+s_{I,J,\ell,m}(\db_K)\right)
\nonumber \\
&\hspace*{3cm}
+\sum_{I=1}^K\sum_{\substack{
\ell,m=0\\
\ell+m\le d_I-2}}^{d_I-1}
{\mathcal M}_{g,k-1,l+2}\left(\bb,\mb+s_{I,\ell,m}(\db_K)\right)
\biggr].
\label{p1}
\end{align}

In the second case
(see Figure \ref{comb2}), the removal changes the numbers $k$ and $n$ to $k-1$ and $n+1$ and the genus of the partial chord diagram decreases by one. 
For partial chord diagram with a boundary with marked point spectrum
\begin{align}
&(d_1,\ldots,d_{I-1},d_I-\ell-1,m,f_{J+1},\ldots,f_L,f_1,\ldots,f_{J-1},f_{J}-m-1,\ell,d_{I+1},\ldots,d_{K}),
\nonumber \\
&0\le \ell \le d_I-1,\quad 0\le m\le f_{J}-1,
\nonumber
\end{align}
we remove the chord which connects the two clusters $(f_{J}-m-1,\ell)$ and $(d_I-\ell-1,m)$ of marked points. The boundary component then splits into two boundary components  with marked point spectra $\db_K=(d_1,\ldots,d_K)$ and  $\fb_L=(f_1,\ldots,f_L)$.
If the original partial chord diagram has the boundary length and point spectrum $\mb$, after removal of this chord, we find that
\begin{align}
&\mb- \pmb{e}_{(d_1,\ldots, d_{I-1},d_{I}-\ell-1,m,f_{J+1},\ldots, f_L,f_1,\ldots,f_{J-1},f_{J}-m-1, \ell, d_{I+1},\ldots, d_K)}+ \pmb{e}_{\db_K} +\pmb{e}_{\fb_L}
\nonumber \\
&
=\mb+q_{I,J,\ell,m}(\db_K,\fb_L).
\nonumber
\end{align}

The number of possibilities of this removal is $(m_{\db_K}+1)(m_{\fb_L}+1)$ for $\db_K\ne\fb_L$. 
If  $\db_K=\fb_L$, the number of possibilities becomes $m_{\db_K}(m_{\db_K}+1)/2$.
In total, the number of possibilities for the second way of elimination is
\begin{align}
&\frac{1}{2}\sum_{K=1}^{\infty}\sum_{L=1}^{\infty}\sum_{\db_K}\sum_{\fb_L}
(m_{\db_K}+1)(m_{\fb_L}+1-\delta_{\db_K,\fb_L})
\nonumber \\
&\quad\times\sum_{I=1}^K\sum_{J=1}^L\sum_{\ell=0}^{d_I-1}\sum_{h=0}^{f_J-1}
{\mathcal M}_{g-1,k-1,l+2}\left(\bb,\mb+q_{I,J,\ell,h}(\db_K,\fb_L)\right).
\label{p2}
\end{align}
The factor $1/2$ in front of the sum takes care of the over counting in the cases $\db_K\ne\fb_L$.

\begin{figure}[h]
\begin{center}
   \includegraphics[width=120mm,clip]{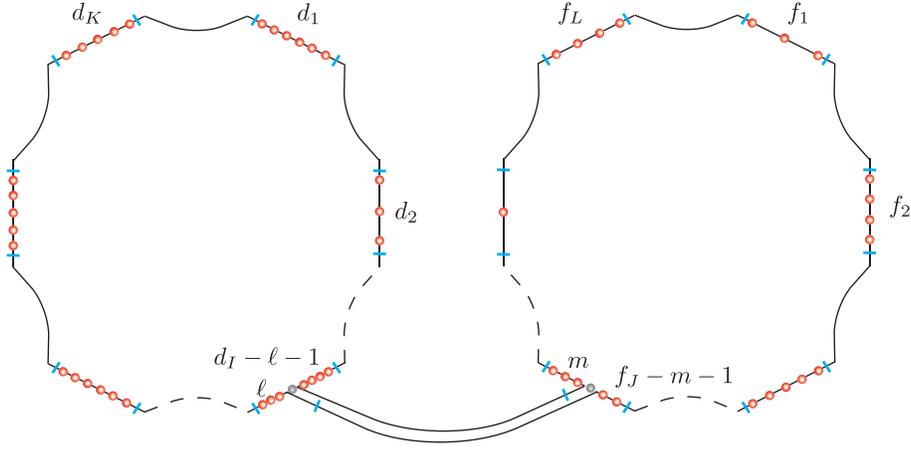}
\end{center}
\caption{\label{comb2} The second and third way of elimination of a chord.
After the elimination of this chord, a boundary component split into  two different boundary components.  }
\end{figure}

In the third case, the partial chord diagram split into two connected components.
We consider the case that the original diagram has the type $\{g,k,l;\bb,\mb\}$ and the resulting two connected components have types $\{g_1,k_1,l_1;\bb^{(1)},\mb^{(1)}\}$ and $\{g_2,k_2,l_2;\bb^{(2)},\mb^{(2)}\}$.
These types are related such that
\begin{align}
g=g_1+g_2,\quad k-1=k_1+k_2,\quad \bb=\bb^{(1)}+\bb^{(2)}.
\nonumber
\end{align}
Since a boundary component also split into two components,
the boundary length and point spectrum changes in the same manner as in the second case.
\begin{align}
\mb+q_{I,J,\ell,m}(\db_K,\fb_L)=\mb^{(1)}+\mb^{(2)}.
\nonumber
\end{align}

There are $m_{\db_K}^{(1)}m_{\fb_L}^{(2)}$ ways to choose the boundary components which are to be fused under the inverse operation of chord removal. And the number of different ordered splittings of a $b$-backbone diagram is $\frac{b!}{b^{(1)}!b^{(2)}!}$ where $b^{(a)}=\sum_ib^{(a)}_i$ ($a=1,2$). Therefore, the total number of possibilities of this case is 
\begin{align}
&\frac{1}{2}\sum_{I=1}^{\infty}\sum_{J=1}^{\infty}\sum_{\db_K}\sum_{\fb_L}
\sum_{g_1+g_2=g}\;\;\sum_{k_1+k_2=k-1}\;\;
\sum_{b^{(1)}+b^{(2)}=b}
\nonumber \\
&\quad\times
\sum_{I=1}^K\sum_{J=1}^L\sum_{\ell=0}^{d_I-1}\sum_{m=0}^{f_J-1}
\sum_{
\substack{
\mb^{(1)}+\mb^{(2)} \\
=\mb+q_{I,J,\ell,m}(\db_K,\fb_L)
}}
\nonumber \\
&\quad\quad\times
m_{\db_K}^{(1)}m_{\fb_L}^{(2)}
\frac{b!}{b^{(1)}!b^{(2)}!}
{\mathcal M}_{g_1,k_1,l_1}\bigl(\bb^{(1)},\mb^{(1)}\bigr)
{\mathcal M}_{g_2,k_2,l_2}\bigl(\bb^{(2)},\mb^{(2)}\bigr).
\label{p3}
\end{align}
The factor $1/2$ corrects for the over counting due to the ordering of the two connected components.

The sum of the contributions  (\ref{p1}), (\ref{p2}), and (\ref{p3}) from the three different cases of chord removals equals $k{\mathcal M}_{g,k,l}(\bb,\mb)$, because  there are $k$ possibilities for the choice of the chord to be removed.
This gives the cut-and-join equation (\ref{cut_join}).
\end{proof}

\begin{proposition}\label{prop2}
The generating function $H(x, y;\ttt,;\uuu)$ is uniquely determined by the differential equation
\begin{align}
{{\partial H}\over{\partial y}} = (M+S)H,
\nonumber
\end{align}
where $M=M_0+x^2M_2$. 
The generating function $Z(x, y;\ttt,;\uuu)=\mathrm{exp}[H]$ of the number of connected and disconnected partial chord diagrams
satisfies 
\begin{align}
\frac{\partial Z}{\partial y}= MZ,
\end{align}
and is as such determined by the initial conditions
\begin{align}
H(x,y=0;\ttt;\uuu)=\sum_{i\geq 1} t_{i}u_{(i)}, \quad Z(x, y=0;\ttt,;\uuu)=\mathrm{e}^{\sum_{i\geq 1} t_{i}u_{(i)}}.
\nonumber
\end{align}
\end{proposition}
\begin{proof}
It is straightforward to check that the differential equation ${{\partial H}\over{\partial y}} = (M+S)H$ is equivalent to
the cut-and-join equation (\ref{cut_join}). The actions in the quadratic differential $S$ on $H$ can be rewritten by
following relation
\begin{align}
D_{\db_K}(H)D_{\fb_L}(H)+D_{\db_K}D_{\fb_L}H=\frac{1}{Z}D_{\db_K}D_{\fb_L}Z.
\nonumber 
\end{align}
The derivatives on the right hand side are contained in $M_2$, and the differential equation $\frac{\partial Z}{\partial y}= MZ$ follows from that of $H$.  
 
On the initial condition, every partial chord diagram of type $\{g, k, l;\bb; \mb\}$ can be obtained from the disjoint collection of type $\{0, 0, i; \pmb{e}_i, \pmb{e}_{(i)}\}$ with multiplicity $b_i$ by connecting them with $k$ chords. This implies $H(x,y=0;\ttt;\uuu)=\sum_{i\geq 1} t_{i}u_{(i)}$. Since this is the first order differential equation of $y$, the coefficient of $y^k$ is determined uniquely using this initial condition.

\end{proof}

%%%%%%%%%%%%%%%%%%%%%%%%%%%%%%%%%%%%%%%%%%%%%%%
%Section 3: Cut-and-join equation for the non-orientable diagrams
%%%%%%%%%%%%%%%%%%%%%%%%%%%%%%%%%%%%%%%%%%%%%%%

\section{Non-oriented analogue of the cut-and-join equation}\label{sec3}
In this section, we will prove Theorem \ref{thm2}. We first establish the following proposition.

\begin{proposition}\label{prop3}
The number $\widetilde{\mathcal M}_{g,k,l}(\bb,\mb)$ of connected non-oriented partial chord diagrams of type $\{g,k,l;\bb,\mb\}$ obeys the following recursion relation
\begin{align}
&k\widetilde{\mathcal M}_{g,k,l}(\bb,\mb)
\nonumber \\
=&\sum_{K\ge1}\sum_{\db_K}
(m_{\db_K}+1)
\nonumber \\
&\times
\Biggl[
\sum_{I<J}\sum_{\ell=0}^{d_I-1}\sum_{m=0}^{d_J-1}
\Bigl\{
\widetilde{\mathcal M}_{h,k-1,l+2}\left(\bb,\mb+s_{I,J,\ell,m}(\mb)\right)
+\widetilde{\mathcal M}_{h-1,k-1,l+2}\bigl(\bb,\mb+s^{\times}_{I,J,\ell,m}(\mb)\bigr)
\Bigr\}
\nonumber \\
&\;
+\sum_{I=1}^K\sum_{\ell+m\le d_I-2}
\Bigl\{
\widetilde{\mathcal M}_{h,k-1,l+2}\left(\bb,\mb+s_{I,\ell,m}(\mb)\right)
+\widetilde{\mathcal M}_{h-1,k-1,l+2}\bigl(\bb,\mb+s^{\times}_{I,\ell,m}(\mb)\bigr)
\Bigr\}
\biggr]
\nonumber \\
&+\frac{1}{2}\sum_{K\ge 1}\sum_{L\ge 1}\sum_{\db_K}\sum_{\fb_L}
(m_{\db_K}+1)(m_{\fb_L}+1-\delta_{\db_K,\fb_L})
\nonumber \\
&\;
\times\sum_{I=1}^K\sum_{J=1}^L\sum_{\ell=0}^{d_I-1}\sum_{m=0}^{f_J-1}
\Bigl\{
\widetilde{\mathcal M}_{h-2,k-1,l+2}\left(\bb,\mb+q_{I,J,\ell,m}(\db_K,\fb_L)\right)
\nonumber \\
&\quad\;
\hspace*{3cm}
+\widetilde{\mathcal M}_{h-2,k-1,l+2}\bigl(\bb,\mb+q^{\times}_{I,J,\ell,m}(\db_K,\fb_L)\bigr)
\Bigr\}
\nonumber \\
&+\frac{1}{2}\sum_{K\ge 1}\sum_{L\ge 1}\sum_{\db_K}\sum_{\fb_L}
\sum_{h_1+h_2=h}\sum_{k_1+k_2=k-1}
\sum_{b^{(1)}+b^{(2)}=b}
\nonumber \\
&\quad\times
\sum_{I=1}^K\sum_{J=1}^L\sum_{\ell=0}^{d_I-1}\sum_{m=0}^{f_J-1}
\vspace{0.3cm}
\left(
\sum_{
\substack{
\mb^{(1)}+\mb^{(2)} \\
=\mb+q_{I,J,\ell,m}(\db_K,\fb_L)
}}
+\sum_{
\substack{
\mb^{(1)}+\mb^{(2)} \\
=\mb+q^{\times}_{I,J,\ell,m}(\db_K,\fb_L)
}}
\right)
m_{\db_K}^{(1)}m_{\fb_L}^{(2)}
\nonumber \\
&\quad\times
\frac{b!}{b^{(1)}!b^{(2)}!}
\widetilde{\mathcal M}_{h_1,k_1,l_1}\bigl(\bb^{(1)},\mb^{(1)}\bigr)
\widetilde{\mathcal M}_{h_2,k_2,l_2}\bigl(\bb^{(2)},\mb^{(2)}\bigr).
\label{cut_join_non}
\end{align}
\end{proposition}

\begin{proof}
If we remove a non-twisted chord, then we find the same recursive structure as for the numbers (\ref{p1}), (\ref{p2}), and (\ref{p3}) for $\widetilde{\mathcal M}_{h,k,l}\bigl(\bb,\mb\bigr)$ in the oriented case.
As we did in the proof of proposition \ref{prop1}, we also consider three cases, organised the same way, when removing a twisted chord.

In the first case (see Figure \ref{comb1_non}),
there are again two possibilities, namely the twisted chord ends belong to two different or the same clusters of marked points on the boundary component in the resulting diagram after removal.
Contrary to the case of non-twisted chords, the boundary cycle does not split, but the marked point spectrum changes due to the recombination of the boundary component.
For both of these two cases, the numbers $k$ and $n$ change to $k-1$ and $n$, and the cross-cap number $h$ decreases by one under this elimination (c.f. Euler's relation $2-h=b-k+n$). The chord ends become marked points and $l$ changes to $l+2$.

\begin{figure}[h]
\begin{center}
   \includegraphics[width=120mm,clip]{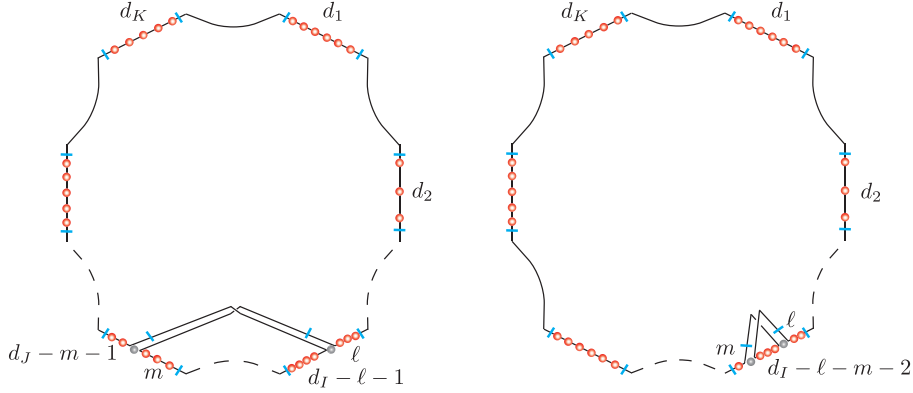}
\end{center}
\caption{\label{comb1_non} Removal of a twisted chord from a non-oriented partial chord diagram. The chord is depicted as a twisted band. After the elimination of this chord, the boundary component is reconnected into one component with different marked point spectrum.
 Left: The clusters of marked points $(d_I-\ell-1,m)$ and $(d_J-m-1,\ell)$  join into two clusters $d_I$ and $d_J$. Right:  The clusters of marked points $(\ell,m)$ and $(d_J-m-1)$  join into one cluster $d_I$.}
\end{figure}

In the former situation, we must have a boundary component with the marked point spectrum 
\begin{align}
&(d_1, \ldots, d_{I-1}, \ell,m, d_{J-1}, d_{J-2}\ldots, d_{I+1}, d_I -\ell-1, d_J -m-1, d_{J+1}, \ldots, d_K)
\nonumber \\
&I<J,\quad 0\le \ell \le d_I-1,\quad 0\le m \le d_J-1,
\nonumber
\end{align}
from which we remove one twisted chord with one end between the two clusters $(\ell,m)$ and the other between $(d_I -\ell-1, d_J -m-1)$.
Then the removal will result in a boundary component with the marked point spectrum $\db_K$ and the boundary length and point spectrum $\mb$ is changed as follows
\begin{align}
&\mb - \pmb{e}_{(d_1, \ldots, d_{I-1}, \ell,m, d_{J-1}, \ldots, d_{J+1}, d_J -\ell-1, d_J -m-1, d_{J+1}, \ldots, d_K)}+ \pmb{e}_{\db_K}
\nonumber \\
&=\mb+s^{\times}_{I,J,\ell,m}(\db_K).
\nonumber
\end{align}
The possible number of choices for this kind of removal is $m_{\db_K}+1$, and the total number of diagrams which can be obtained in this way is
\begin{align}
&\sum_{K\ge1}\sum_{\db_K}
(m_{\db_K}+1)
\sum_{I<J}\sum_{\ell=0}^{d_I-1}\sum_{m=0}^{d_J-1}
\widetilde{\mathcal M}_{h-1,k-1,l+2}\bigl(\bb,\mb+s^{\times}_{I,J,\ell,m}(\mb)\bigr).
\label{p1_non}
\end{align}

For the removal of the latter kind of twisted chords, 
we must start with a diagram with a boundary component with the marked point spectrum 
\begin{align}
&(d_1,\ldots, d_{I-1}, \ell, d_I-\ell-m-2, m, d_{I+1},\ldots, d_K),
\nonumber \\
&0\le \ell,m\le d_I,\quad \ell+m\le d_I-2.
\nonumber
\end{align}
from which we remove one twisted chords with one end between the two clusters $(\ell,d_I-\ell-m-2)$ and the other one between the two clusters $(d_I-\ell-m-2,m)$. After removal, we obtain a boundary component with the marked point spectrum $\db_K$. Thus,  the boundary length and point spectrum $\mb$ is changed to
\begin{align}  
\mb- \pmb{e}_{(d_1,\ldots, d_{I-1}, \ell, d_I-\ell-m-2, m, d_{I+1},\ldots, d_K)}+ \pmb{e}_{\db_K}
=\mb+s^{\times}_{I,\ell,m}(\db_K).
\nonumber
\end{align}
The number of such chords to be removed is $m_{\db_K}+1$, and the total number of partial chord diagrams obtained in this way is
\begin{align}
\sum_{K\ge1}\sum_{\db_K}
(m_{\db_K}+1)
\sum_{I=1}^K\sum_{\ell+m\le d_I-2}
\widetilde{\mathcal M}_{h-1,k-1,l+2}\bigl(\bb,\mb+s^{\times}_{I,\ell,m}(\mb)\bigr).
\label{p2_non}
\end{align}

Next, we consider the second case (see Figure \ref{comb2_non}), where we must start with a non-oriented partial chord diagram with a boundary component with the marked point spectrum
\begin{align}
&(f_{1},\ldots, f_{J-1},f_{J}-m-1,\ell, d_{I-1},\ldots, d_1,d_K,\ldots, d_{I+1}, d_I -\ell-1,m, f_{J+1},\ldots, f_L),
\nonumber \\
&0\le \ell\le d_I-1,\quad 0\le m\le f_J-1,
\nonumber
\end{align}
from which we remove a twisted chord with one end between the two clusters $(f_{J}-m-1,\ell)$ and the other end between the two clusters $(d_I -\ell-1,m)$.
After removal of this chord, the boundary component has been split into two components with spectra $\db_K$ and $\fb_L$, and the cross-cap number $h$ decreases by two.
Then, the boundary length and point spectrum $\mb$ changes to
\begin{align}
&\mb- \pmb{e}_{(f_{1},\ldots, f_{J-1},f_{J}-m-1,\ell, d_{I-1},\ldots, d_1,d_K,\ldots, d_{I+1}, d_I -\ell-1,m, f_{J+1},\ldots, f_L)}+ \pmb{e}_{\db_K}  +\pmb{e}_{\fb_L}
\nonumber \\
&=\mb+q^{\times}_{I,J,\ell,m}(\db_K,\fb_L).
\nonumber
\end{align}
The number of choices for the chord to be removed is  $(m_{\db_K}+1)(m_{\fb_L}+1)$ for $\db_K\ne \fb_L$ and $m_{\db_K}(m_{\db_K}+1)/2$ for $\db_K= \fb_L$, and the total number of partial chord diagrams obtained this way is
\begin{align}
&\frac{1}{2}\sum_{K\ge 1}\sum_{L\ge 1}\sum_{\db_K}\sum_{\fb_L}
(m_{\db_K}+1)(m_{\fb_L}+1-\delta_{\db_K,\fb_L})
\nonumber \\
&\;
\times\sum_{I=1}^K\sum_{J=1}^L\sum_{\ell=0}^{d_I-1}\sum_{m=0}^{f_J-1}
\widetilde{\mathcal M}_{h-2,k-1,l+2}\bigl(\bb,\mb+q^{\times}_{I,J,\ell,m}(\db_K,\fb_L)\bigr).
\label{p3_non}
\end{align}

\begin{figure}[h]
\begin{center}
   \includegraphics[width=120mm,clip]{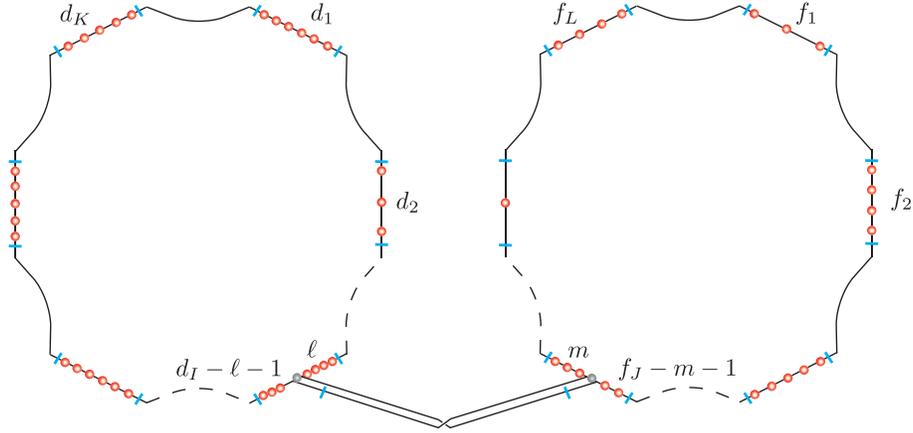}
\end{center}
\caption{\label{comb2_non} 
The second case of a twisted chord removal.
After the removal of this chord, the boundary component split into two distinct boundary components. }
\end{figure}

In case three partial chord diagram split into two connected components when we remove the chord.
Assume that the original diagram has the type $\{h,k,l;\bb,\mb\}$ and the resulting two connected components have types $\{h_1,k_1,l_1;\bb^{(1)},\mb^{(1)}\}$ and $\{h_2,k_2,l_2;\bb^{(2)},\mb^{(2)}\}$.
Then these types are related by
\begin{align}
h=h_1+h_2,\quad k-1=k_1+k_2,\quad \bb=\bb^{(1)}+\bb^{(2)}.
\nonumber
\end{align}
The marked point spectrum changes in the same way as the second case %in case two
\begin{align}
\mb+q^{\times}_{I,J,\ell,m}(\db_K,\fb_L)=\mb^{(1)}+\mb^{(2)}.
\nonumber
\end{align}
The total number of resulting diagrams is
\begin{align}
&\frac{1}{2}\sum_{K\ge 1}\sum_{L\ge 1}\sum_{\db_K}\sum_{\fb_L}
\sum_{h_1+h_2=h}\sum_{k_1+k_2=k-1}
\sum_{b^{(1)}+b^{(2)}=b}
\nonumber \\
&\quad\times
\sum_{I=1}^K\sum_{J=1}^L\sum_{\ell=0}^{d_I-1}\sum_{m=0}^{f_J-1}
\vspace{0.3cm}
\sum_{
\substack{
\mb^{(1)}+\mb^{(2)} \\
=\mb+q^{\times}_{I,J,\ell,m}(\db_K,\fb_L)
}}
\nonumber \\
&\quad\times
m_{\db_K}^{(1)}m_{\fb_L}^{(2)}
\frac{b!}{b^{(1)}!b^{(2)}!}
\widetilde{\mathcal M}_{h_1,k_1,l_1}\bigl(\bb^{(1)},\mb^{(1)}\bigr)
\widetilde{\mathcal M}_{h_2,k_2,l_2}\bigl(\bb^{(2)},\mb^{(2)}\bigr).
\label{p4_non}
\end{align}

Therefore, in total, the number of possible partial chord diagrams obtained by removing a twisted or a  non-twisted  chord is the sum of (\ref{p1_non}) -- (\ref{p4_non}) and  of (\ref{p1}) -- (\ref{p3}) for $\widetilde{\mathcal M}_{h,k,l}\bigl(\bb,\mb\bigr)$.
This number gives the right hand side of equation (\ref{cut_join_non}), which we have just argued also gives the left side of equation (\ref{cut_join_non}).

\end{proof}

Along the same line of arguments as the ones which proved Proposition \ref{prop2}, we obtain the proposition below.

\begin{proposition}\label{prop4}
The generating function $\widetilde{H}(x, y;\ttt,;\uuu)$ is uniquely determined by the differential equation
\begin{align}
{{\partial \widetilde{H}}\over{\partial y}} = (\widetilde{M}+\widetilde{S})\widetilde{H},
\nonumber
\end{align}
where $\widetilde{M}=M_0+xM_1^{\times}+x^2(M_2+M_2^{\times})$ and $\widetilde{S}=S+S^{\times}$. 
The generating function $\widetilde{Z}(x, y;\ttt,;\uuu)=\mathrm{exp}[\widetilde{H}]$ of the number of connected and disconnected partial chord diagrams filtered by the boundary length and point spectrum
satisfies 
\begin{align}
\frac{\partial \widetilde{Z}}{\partial y}= \widetilde{M}\widetilde{Z}.
\end{align}
As such they are uniquely determined by the initial conditions 
\begin{align}
\widetilde{H}(x,y=0;\ttt;\uuu)=\sum_{i\geq 1} t_{i}u_{(i)}, \quad \widetilde{Z}(x, y=0;\ttt,;\uuu)=\mathrm{e}^{\sum_{i\geq 1} t_{i}u_{(i)}}.
\nonumber
\end{align}
\end{proposition}

%%%%%%%%%%%%%%%%%%%%%%%%%%%%%%%%%%%%%%%%%%%%%%%%%%%%%%%%%%%%%%%%%%%%%%%
%%%%%%%%%%%%%%%%%%%%%%%%%%%%%%%%%%%%%%%%%%%%%%%%%%%%%%%%%%%%%%%%%%%%
%%%%%          The bibliography
%%%%%%%%%%%%%%%%%%%%%%%%%%%%%%%%%%%%%%%%%%%%%%%%%%%%%%%%%%%%%%%%%%%%%

\end{document}